\documentclass[11pt]{article}
\usepackage[margin=1in]{geometry}
\pdfoutput=1
\usepackage{enumitem}
\usepackage[utf8]{inputenc}
\usepackage{color}
\usepackage{multirow}
\usepackage{boxedminipage}
\usepackage{xifthen}
\usepackage{tabularx}
\usepackage{graphicx}
\usepackage{amsfonts}
\usepackage{amssymb}
\usepackage{amsmath}
\usepackage{amsthm}
\usepackage{bbm}
\usepackage[ruled,vlined]{algorithm2e}
\usepackage{thmtools}
\usepackage{thm-restate}
\usepackage{color}
\usepackage{wrapfig}
\usepackage{hyperref}
\hypersetup{
    colorlinks=true,
    linkcolor=blue,
    filecolor=magenta,      
    urlcolor=cyan,
    citecolor=red,
    pdftitle={CHAR},
    pdfpagemode=FullScreen,
    }
\allowdisplaybreaks

\newtheorem{theorem}{Theorem}

\newtheorem{definition}{Definition}
\newtheorem{lemma}{Lemma}
\newtheorem{remark}{Remark}

\begin{document}
\global\long\def\norm#1{\left\Vert #1\right\Vert }%
\global\long\def\R{\mathbb{R}}%
 
\global\long\def\Rn{\mathbb{R}^{n}}%
\global\long\def\tr{\mathrm{Tr}}%
\global\long\def\diag{\mathrm{diag}}%
\global\long\def\Diag{\mathrm{Diag}}%
\global\long\def\C{\mathbb{C}}%
 
\global\long\def\E{\mathbb{E}}%
\global\long\def\vol{\mathrm{vol}}%
\global\long\def\argmax{\mathrm{argmax}}%

\global\long\def\ham{\mathrm{Ham}}%
\global\long\def\e#1{ \exp\left(#1\right)}%
\global\long\def\Var{\mathrm{Var}}%
\global\long\def\dint{{\displaystyle \int}}%
\global\long\def\step{\delta}%
\global\long\def\Ric{\mathrm{Ric}}%
\global\long\def\P{\mathbb{P}}%
\global\long\def\len{\text{\text{len}}}%
\global\long\def\lspan{\mathrm{span}}%

\title{Convergence of Gibbs Sampling:\\ Coordinate Hit-and-Run Mixes Fast
}
\author{Aditi Laddha\thanks{Georgia Tech, \{aladdha6, vempala\}@gatech.edu}\and
Santosh Vempala\footnotemark[1]}
\date{}
\maketitle

\begin{abstract}
Gibbs sampling, also known as Coordinate Hit-and-Run (CHAR), is a Markov chain Monte Carlo algorithm for sampling from high-dimensional
distributions. In each step, the algorithm selects a random coordinate and re-samples that coordinate from the distribution induced by fixing all the other coordinates. While this algorithm has become widely used over the past half-century, guarantees of efficient convergence have been elusive. We show that the Coordinate Hit-and-Run algorithm for sampling from a convex body $K$ in $\mathbb{R}^{n}$ mixes in $O^{*}(n^9 R^2/r^2)$ steps, where $K$ contains a ball of radius $r$ and $R$ is the average distance of a point of $K$ from its centroid. We also give an upper bound
on the conductance of Coordinate Hit-and-Run, showing that it is strictly worse than Hit-and-Run
or the Ball Walk in the worst case.
\end{abstract}

\section{Introduction}
\label{intro}
Sampling from a distribution in high-dimensional space is a fundamental problem 
and an essential ingredient of algorithms for optimization, integration,
statistical inference, and other applications. Progress on sampling
algorithms has led to many useful tools, both theoretical and practical.
In the most general setting, given access to a function $f:\R^{n}\rightarrow\R_{+},$
the goal is to generate a point $x$ whose density is proportional
to $f(x)$. Two special cases of particular interest are when $f$
is uniform over a convex body and when $f$ is a Gaussian restricted to a convex
set.

The general approach to sampling is to design an ergodic, time-reversible Markov chain whose state space is the convex body and which has the desired density as its stationary distribution. The key
question is to bound the rate of convergence of the Markov chain. The Ball walk~\cite{L90,KLS97,LV1} and Hit-and-Run~\cite{Boneh83,Smith84,LV3} are two Markov chains that
work in full generality and have been shown to mix rapidly (i.e.,
the convergence rate is polynomial in ambient dimension) for arbitrary log-concave densities.
Three decades of improvements have reduced the complexity of this problem to a small polynomial in the dimension. For a log-concave
density with support of diameter $D$, the mixing time of both Ball Walk and Hit-and-Run is $O^{*}(n^{2}D^{2})$\footnote{The $O^*$ notation suppresses
logarithmic factors and dependence on other parameters like error bound}, each step requiring $O(n^2)$ function evaluations and $O(n^2)$ arithmetic operations, giving total arithmetic complexity of $O^{*}(n^{4}D^{2})$\cite{KLS97,Lovasz1998,LV3}.

A simple and widely-used algorithm that pre-dates these developments is the Gibbs Sampler, proposed by Turchin in 1971 \cite{Turchin1971}.
It is inspired by statistical physics and is commonly used for sampling distributions \cite{diaconis2010gibbs,diaconis2012gibbs} and Bayesian
inference \cite{finkel2005incorporating,DBLP:journals/pami/GemanG84,george1993variable}. To sample from a multivariate density, at each step, the Gibbs sampling algorithm
selects a coordinate (either at random or in order, cycling through
the coordinates), fixes all other coordinates, and re-samples this coordinate from the induced distribution. This algorithm is very similar to
Hit-and-Run, except that instead of picking a direction uniformly
at random from the unit sphere, it is picked only from one of the
$n$ basis vectors (see \cite{Anderson2007} for a historical account
and more background). It was reported to be significantly faster than
Hit-and-Run in state-of-the-art software for volume computation and
integration \cite{CV13,emiris2014efficient,CV15b}. Gibbs sampling, also called Coordinate
Hit-and-Run, has a computational benefit: updating the current point
takes $O(n)$ time rather than $O(n^{2})$ even for polyhedra since
the update is along only one coordinate direction. Thus the overhead
per step is reduced from $O(n^{2})$, as in all previous algorithms,
to $O(n)$. However, despite half a century of intense study, the
convergence rate of Gibbs sampling has remained an open problem. 

This paper shows that Gibbs sampling/Coordinate Hit-and-Run mixes rapidly for any
convex body $K$. Before stating our main theorem formally, we define the Coordinate Hit-and-Run.

\paragraph{Coordinate Hit-and-Run.} Algorithm~\ref{alg:char} describes the Coordinate Hit-and-Run Markov chain, hereafter referred to as CHAR, for sampling uniformly from a convex body
$K\in\mathbb{R}^{n}$. Let $\{e_{i}: i\in [n]\}$ be
the standard basis for $\mathbb{R}^{n}.$ The input to the algorithm is the convex body $K$, a starting point $x^{(0)}$ in
the interior of $K$, and the number of steps $T$.

\begin{algorithm}[ht]

\caption{Coordinate Hit-and-Run (CHAR)}
\label{alg:char}
\SetAlgoLined

\textbf{Input:} a point $x^{(0)}\in K$, integer $T$.

\For{$i=1,2,\cdots,T$}{

Pick a uniformly random axis direction $e_{j}$

Set $x^{{(i)}}$ to be a random point along the line $\ell=\left\{ x^{(i-1)}+te_{j}\,:\:t\in\R\right\} $
chosen uniformly from $\ell\cap K$$.$

}
\textbf{Output:} $x^{{(T)}}$.

\end{algorithm}
 The stationary distribution of the Coordinate hit-and-run walk is the uniform distribution $\pi_K$ over $K$. To sample from a general log-concave density $f:\R^{n}\rightarrow\R_{+}$
the only change is in Step 2, where the next point $y$ is chosen
according to $f(y)$ restricted to $\ell$. In both cases, the process
is symmetric and ergodic, so the stationary distribution of the
Markov chain is the desired distribution.

We can now state our main theorem (see Sec.~\ref{sec:prelims} for the definition of a warm start).
\begin{theorem}
\label{thm:mixing}Let $K$ be a convex body in $\R^{n}$ containing
a unit ball. Let $R^{2}$ be the expected squared distance of a uniformly random point in $K$ from the centroid of $K$, and let $\pi_K$ denote the uniform distribution on $K$. Let $\sigma$ be a starting distribution and let $
\sigma^m$ be the distribution of the current point
after $m$ steps of Coordinate Hit-and-Run in $K$. Let $\varepsilon > 0$, and suppose that $\sigma$ is $M$-warm with respect to $\pi_K$. Then for
\begin{equation*}
    m > 7\cdot 10^4 \cdot \frac{M^2R^2n^9}{\varepsilon^2}\log\left(\frac{2M}{\varepsilon}\right),
\end{equation*}
the total variation distance between $\sigma^m$ and $\pi_K$ is less than $\varepsilon$.
\end{theorem}

By applying an affine transformation, $R$ can be made $O(\sqrt{n})$ (see \cite{CV2015}, \cite{KLS97}, \cite{LV2}). 
We note that from
a warm start, both the Ball Walk and Hit-and-Run have a mixing time of $O^*(n^{2}R^{2})$ \cite{KLS97,LV3}. While
our bound is likely not the best polynomial bound for CHAR, in Section
\ref{sec:Lower-bound}, we show that it is necessarily higher than
the bound for Hit-and-Run.

Concurrently and independently, Narayanan and Srivastava~\cite{narayanan2020mixing} also proved a polynomial bound on the mixing rate of Coordinate Hit-and-Run, with a different proof. They showed that CHAR mixes in $O^*(n^7R_1^4)$ steps where $R_1$ is the smallest number s.t., $B_\infty \subseteq K \subseteq R_1 B_\infty$, i.e., the cube sandwiching ratio ($R_1$ can be larger than $R$ in our theorem by a factor of $\sqrt{n}$). After an affine transformation, $R_1$ can be bounded by $O(n)$.

A key ingredient of our proof is a new ``$\ell_{0}$''-isoperimetric
inequality. We will need the following definition.
\begin{definition}[Axis-disjoint]
 Two measurable sets $S_{1},S_{2}$ are called \emph{axis-disjoint}
if $\forall x\in S_{1},\forall y\in S_{2},\vert\{i\in\left[n\right]:x_{i}=y_{i}\}\vert\leq n-2$.
\end{definition}
In other words, it is not possible to move from a point in $S_{1}$ to any point in $S_{2}$ in one step of CHAR and vice versa (see Fig. \ref{fig:axis-disjoint}).

\begin{figure}[ht]
\centering
\includegraphics[width = 0.5\textwidth]{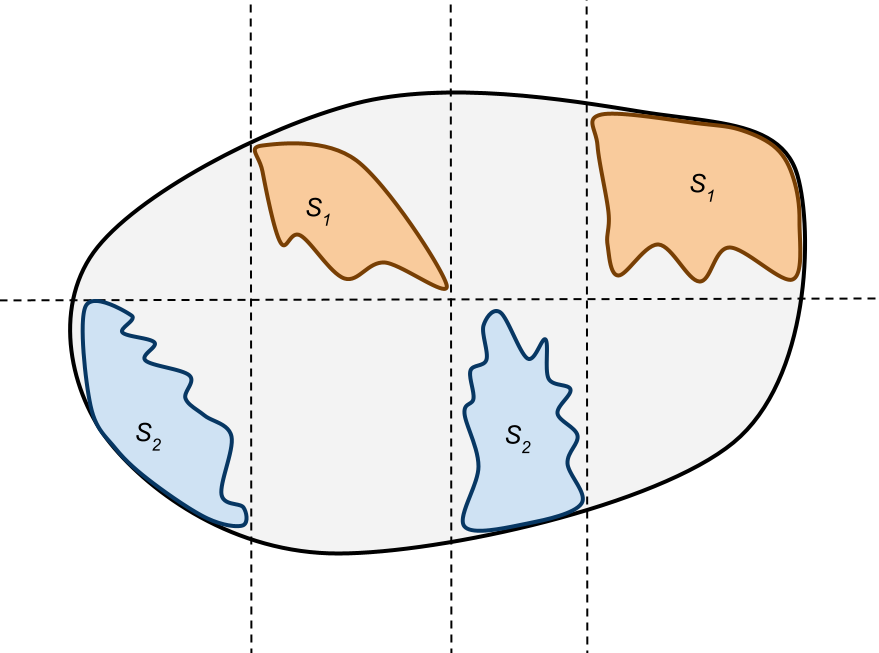}
\caption{Axis-disjoint subsets $S_{1}$ and $S_{2}$.}
\label{fig:axis-disjoint}
\end{figure}

The main component of the proof of Theorem \ref{thm:mixing} is the following isoperimetric inequality for axis-disjoint subsets of a convex body.
\begin{restatable}{theorem}{isoperimetry}
\label{thm:iso}Let $K$ be a convex body in $\R^{n}$ containing
a unit ball with $R^{2}=\E_{K}(\norm{x-z_K}^{2})$ where $z_K$ is the centroid of $K$. Let $S_{1,}S_{2}\subseteq K$
be two measurable subsets of $K$ such that $S_{1}$ and $S_{2}$ are axis-disjoint.
Then for any $\varepsilon\ge0$, the set $S_{3}=K\backslash \{S_{1} \cup S_{2}\}$
satisfies
\begin{equation*}
    \vol(S_{3})\geq\frac{\varepsilon}{800 \cdot n^{3.5}R}\cdot \left(\min\{\vol(S_{1}), \vol(S_{2})\}-\varepsilon \cdot \vol(K)\right).
\end{equation*}
\end{restatable}

At a high level, we follow the proof of rapid mixing based on the
conductance of Markov chains \cite{SJ89} in the continuous setting
\cite{LS93}. We give a simple, new one-step coupling lemma (Lemma~\ref{lem:disjoint}) which reduces the problem of lower bounding the conductance of CHAR for a convex body $K$ to lower bounding the isoperimetric coefficient of axis-disjoint subsets
in $K$. Roughly speaking, the inequality says the following:
If two subsets of $K$ are axis-disjoint, then the remaining
mass of the body is proportional to the smaller of the two subsets.
This inequality is our main technical contribution. In comparison, the isoperimetric inequality for Euclidean distance says that for any two subsets of a convex body,
the remaining mass is proportional to their (minimum) Euclidean distance
times the smaller of the two subset volumes.

Standard approaches to proving such inequalities, notably localization
\cite{KLS95,LeeV17KLS}, which reduce the desired high-dimensional
inequality to a one-dimensional inequality, do not seem to be directly
applicable to proving this ``$\ell_{0}$-type'' inequality. So we
develop a first-principles approach where we first prove the isoperimetric inequality
for cubes (Lemma~\ref{lem:cube_isoperimetry}), taking advantage of their product structure, and then for general convex bodies (Theorem~\ref{thm:iso}) using a tiling of space with cubes. In the latter part, we will use several known properties of convex bodies,
including Euclidean isoperimetry.

\subsection{Preliminaries.}\label{sec:prelims}
We restate a few useful definitions from \cite{LS93}.
\begin{itemize}
    \item \textbf{Markov chain}: Let $\mathcal{M}$ be a Markov chain with state space $K$ and stationary distribution $Q$. For any measurable subset $S\subseteq K$ and $x \in K$, let $P_{x}(S)$
be the probability that one step of $\mathcal{M}$ from $x$
goes to a point in $S$. Distribution $Q$ is called stationary if one step from it gives the same distribution, i.e., for any measurable subset $S\subseteq K$,
\[
\int_K P_x(S) \;dQ(x) = Q(S).
\]
The Markov chain is {\bf time-reversible} if for any two subsets $A,B\subseteq K$ 
\[
\int_A P_{x}(B)\;dQ(x)=\int_B P_{x}(A)\;dQ({x}).
\]
In other words, the probability of going from $A$ to $B$ is the same as that of going from $B$ to $A$ for a reversible Markov chain.

The ergodic flow of a measurable subset $S\subseteq K$, denoted by $p(S)$  is defined as
\begin{equation*}
 p(S)=\int_{S}P_{x}(K\setminus S)\,dQ(x).
\end{equation*}
\item \textbf{Conductance}:
The conductance of a subset $S\subseteq K$, denoted by $\phi(S)$,  is defined as
\begin{equation*}
 \phi(S)=\frac{\int_{S}P_{x}(K\setminus S)\,dQ(x)}{\min\{Q(S), Q(K\backslash S)\}},
\end{equation*}
``Conductance measures the probability of moving from $S$ to $K\backslash S$, conditioned on starting in $S$ in the stationary distribution'' \cite{kannan2003rapid}.

The conductance of the Markov chain is defined as
\begin{equation*}
    \phi = \inf_{0 < Q(S) \leq 1/2} \phi(S).
\end{equation*}
For any $s\in[0,1/2]$ the $s$-conductance of the Markov chain is defined as
\[
\phi_{s}=\inf_{S:s < Q(S)\le\frac{1}{2}} \frac{p(S)}{Q(S)-s}.
\]
$s$-conductance is a weaker notion of conductance; it is defined only for sets with relatively large measures. Both conductance and $s$-conductance can be used to bound the mixing time of a Markov chain \cite{LS93}.
\item \textbf{Warm Start}: Given distributions $P$ and $Q$ on the same state space $\mathcal{A}$, $P$ is said to be $M$-warm with respect to $Q$ if 
\begin{equation*}
    M = \sup_{A\subseteq \mathcal{A}} \frac{P(A)}{Q(A)}.
\end{equation*}
If the initial distribution $Q_0$ is $O(1)$-warm with respect to the stationary distribution $Q$, we say that $Q_0$ is a warm start for $Q$.
\item \textbf{Lazy chain}: A lazy version of a Markov chain with transition probability $P$ is one where we use the transition probability $P_x(\{y\}) =
P_x(\{y\})/2 + \mathbf{1}(x=y)/2$. With probability $1/2$, the chain feels lazy and stays in the same state.
\item For a body $K\subseteq \mathbb{R}^n$, let $\pi_K$ denote the uniform distribution on $K$ and $\mathbb{E}_K(X)$ denote the expected value of $X$ with respect to $\pi_K$.
\item For a set $S \subset \mathbb{R}^n$, we use $\partial S$ to denote the boundary of $S$ and $\vol(S)$ to denote its $n$-dimensional Lebesgue measure.
    \item \textbf{Internal Boundary}: For a convex body $K$ and a measurable subset $S \subseteq K$, the internal boundary of $S$ with respect to $K$ is defined as $\partial_K(S) = \partial S \cap {\mathrm{Int}} K$, where $\mathrm{Int} K$ denotes the interior of $K$.
\end{itemize}

The following theorem shows that the $s$-conductance of a Markov
chain bounds its rate of convergence from a warm start.
\begin{theorem}
\cite{LS93} \label{thm:sconductance} Suppose that a lazy, time-reversible Markov chain with
stationary distribution $Q$ has $s$-conductance at least $\phi_{s}$.
Then with initial distribution $Q_{0}$, and 
\[
H_{s}=\sup\left\{ \left|Q(A)-Q_{0}(A)\right|\,:\,A\subset K,Q(A)\le s\right\} ,
\]
the distribution $Q_{t}$ after $t$ steps satisfies 
\[
d_{TV}(Q_{t},Q)\le H_{s}+\frac{H_{s}}{s}\left(1-\frac{\phi_{s}^{2}}{2}\right)^{t}.
\]
\end{theorem}
 
\section{The isoperimetric inequality}

\label{sec:isoperimetry}
Before proving Theorem~\ref{thm:iso}, we need a few definitions. 
\begin{definition}[Axis-aligned Line]
A line $\ell$ in $\mathbb{R}^n$ is called \emph{axis-aligned} if $\ell = \{ x_0+te_i: t \in \mathbb{R}\}$ for some $i\in \{1,\ldots, n\}$ and a point $x_0\in \mathbb{R}^n$.
\end{definition}
\begin{definition}[Axis-aligned Cube]
A cube $C$ is called \emph{axis-aligned} if
\begin{equation*}
    C = \{x \in \mathbb{R}^n : \norm{x-x_0}_{\infty} \leq r\},
\end{equation*}
where $x_0\in \mathbb{R}^n$ is a fixed point and $r$ is a positive constant.
\end{definition}
\begin{definition}[Isoperimetric coefficients for cubes]
The isoperimetric coefficient for axis-aligned cubes, $\psi_c$, is the largest positive real such that for any axis-aligned cube $C\in\mathbb{R}^{n}$, and any two axis-disjoint
subsets $S_{1},S_{2}\subseteq C$, with $S_{3}=C\setminus \{ S_{1}\cup S_{2}\}$,
\begin{equation*}
\vol(S_{3})\ge\psi_c \cdot\min\left\{ \vol(S_{1}),\vol(S_{2})\right\} .
\end{equation*}
\end{definition}

\begin{lemma}[Cube isoperimetry\label{lem:Cube-isoperimetry}]
\label{lem:cube_isoperimetry}
Let $C$ be an axis-aligned cube, and let $S_1,S_2 \subseteq C$ be two measurable subsets of $C$ such that $S_1$ and $S_2$ are axis-disjoint. Then the set $S_{3}=C\setminus \{ S_{1}\cup S_{2}\}$ satisfies
\begin{equation*}
\vol(S_{3})\ge\frac{\ln{2}}{n}\cdot\min\left\{ \vol(S_{1}),\vol(S_{2})\right\} .
\end{equation*}
\end{lemma}
\begin{remark}
We believe that the bound above is not optimal, and even an absolute
constant factor might be possible.
\end{remark}
\begin{proof}
Without loss of generality, let $C$ be a unit cube and let $\vol(S_1)\leq \vol(S_2)$. For each coordinate $i \in [n]$, consider the projection functions
\begin{align*}
    \pi_{i}&:{\mathbb  {R}}^{{n}}\to {\mathbb  {R}}^{{n-1}},\\
    \pi _{i}&:x=(x_{1},\dots ,x_{n})\mapsto {\hat {x}}_{i}=(x_{1},\dots ,x_{i-1},x_{i+1},\dots ,x_{n}).
\end{align*}
For each $i \in [n]$, extend $\pi_i(S_1)$ to
\begin{equation*}
    \mathcal{E}_i(S_1) = \{(x,y): x \in \pi_i(S_1), y \in [0,1] \}.
\end{equation*}
Note that for all $i \in [n]$, $\mathcal{E}_i(S_1) \cap S_2 = \emptyset$ as $S_1$ and $S_2$ are axis disjoint. Thus,  $\mathcal{E}_i(S_1) \subseteq S_1 \cup S_3$ and
\begin{equation}
    \vol(\mathcal{E}_i(S_1)) \leq  \vol(S_1) + \vol(S_3).\label{eq:1}
\end{equation}
Since the side length of the cube is $1$, $\vol_n(\mathcal{E}_i(S_1)) = \vol_{n-1}(\pi_i(S_1))$.
Summing up inequality~\eqref{eq:1} over all $i \in [n]$, we get
\begin{align*}
     n\cdot \vol(S_1) &+ n\cdot \vol(S_3) \geq \sum_{i=1}^n\vol(\mathcal{E}_i(S_1)) =  \sum_{i = 1}^n \vol_{n-1}(\pi_i(S_1))
\end{align*}
After shifting $n\cdot \vol(S_1)$ to the RHS and using AM-GM inequality,
\begin{align*}
     n\cdot \vol(S_3) &\geq n\cdot \left( \prod_{i=1}^n \vol_{n-1}(\pi_i(S_1))\right)^{\frac{1}{n}} - n \cdot \vol(S_1).
\end{align*}
The Loomis–Whitney inequality \cite{loomis1949inequality} states that for any subset $S \subset \mathbb{R}^n$,
\begin{equation*}
    \prod_{i=1}^n \vol_{n-1}(\pi_i(S)) \geq \vol(S)^{n-1}.
\end{equation*}
Using the Loomis Whitney inequality on $S_1$, we get
\begin{align*}
    n\cdot \vol(S_3) &\geq  n\cdot \left(  \vol(S_1)\right)^{\frac{n-1}{n}} - n \cdot \vol(S_1) \\
    &= n\cdot \vol(S_1) \cdot \left( \frac{1}{\vol(S_1)^{\frac{1}{n}}}-1\right) \geq n \cdot \vol(S_1) \cdot (2^{\frac{1}{n}} - 1) \\
    &\geq \vol(S_1)\cdot \ln{2}.
\end{align*}
The last inequality follows from the fact that $\lim_{n\to \infty} n(2^{\frac{1}{n}}-1) = \ln 2$. 
\end{proof}
Before proceeding to the proof of isoperimetry for general convex bodies, we state two lemmas:
\begin{lemma} \label{lem:cube-}
Let $C$ be an axis-aligned unit cube. Let $S_1$ and $S_2$ be axis-disjoint subsets of $C$ with $\vol(S_1) \leq \frac{2}{3} \cdot \vol(C)$. Let $S_3 = C\backslash \{S_1 \cup S_2\}$. Then
\begin{equation*}
    \vol(S_3) \geq \frac{\psi_c}{4} \cdot \vol(S_1).
\end{equation*}
\end{lemma}
\begin{proof}
Since $S_3 = C \backslash \{S_1 \cup S_2\}$,
\begin{equation}
   \vol(S_3) \geq 1 - \vol(S_1) - \vol(S_2) \label{eq:2}
\end{equation}
 Applying Lemma \ref{lem:Cube-isoperimetry} to $C$ gives
 \begin{equation}
     \vol(S_3) \geq \psi_c\cdot \min \left\{\vol(S_1),\vol(S_2)\right\} \label{eq:3}
 \end{equation}
A lower bound on $\vol(S_3)$ is the maximum of the bounds obtained from~\eqref{eq:2} and~\eqref{eq:3}, i.e.,
\begin{align*}
   \vol(S_3 ) &\geq \max\left\{1 - \vol(S_1) - \vol(S_2),\; \psi_c\cdot \min \left\{\vol(S_1),\vol(S_2)\right\}\right\} \notag \\
    &\geq \psi_c \cdot \frac{\vol(S_1)}{4}.
\end{align*}
The last inequality follows from the constraints $ 0 < \vol(S_1) \leq 2/3$ and $ \vol(S_1)+\vol(S_2) \leq 1$.
\end{proof}
In the next lemma, we restate an isoperimetric inequality from \cite{KLS95}.

\begin{lemma}[Euclidean isoperimetry\label{Euclidean-isoperimetry}]
\cite{KLS95} Let $K\subset\R^{n}$ be a convex body containing a
unit ball and $R^{2}=\E_{K}(\norm{x-z_k}^{2})$ where $z_K$ is the centroid of $K$. For a subset
$S\subseteq K$, let $\partial_K(S)$ denote the
boundary of $S$, relative to $K$. Then for any $S \subseteq K$  of volume at most $\vol(K)/2$, we have 
\[
\vol_{n-1}(\partial_K S)\ge\frac{\ln 2}{R} \cdot \vol(S).
\]
\end{lemma}

We can now prove the isoperimetric inequality for axis-disjoint subsets of a convex body, which we restate below for convenience.

\isoperimetry*
\begin{proof}
Without loss of generality, let $\vol(S_{1})\le\vol(S_{2})$.
We also assume that $K$ contains the origin (otherwise we can shift $K$ by its mean). Let $K'=(1-\alpha)K$
for $\alpha = \frac{\varepsilon}{2 n}$, and let $S_{i}'=S_{i}\cap K'$ for $i\in \{1,2\}$.
Note that $\vol(K') = (1-\alpha)^n \cdot \vol(K)$ and therefore 
\begin{equation*}
\vol(K\backslash K') = (1-(1-\alpha)^n)\cdot\vol(K) \geq \frac{\epsilon}{2}\cdot \vol(K).
\end{equation*}
For any set $X \subseteq K$, we have
\begin{equation}
    \vol(X \cap K') \geq \vol(X) - \vol(K\backslash K')\geq  \vol(X) - \frac{\varepsilon}{2}\cdot \vol(K). \label{eq:4}
\end{equation}

Next, consider a standard lattice of width $\delta$, with each lattice
point inducing an $n$-dimensional cube of side length $\delta$. Since $K$ contains a unit ball, $\delta=\frac{\alpha}{4\sqrt{n}}$ ensures that any cube that intersects $K'$ and all its neighboring lattice cubes are fully contained in $K$. 
For a set of cubes $\mathcal{W}$, let $\mu(\mathcal{W}) = \cup_{c\in \mathcal{W}} c$.

Let ${\cal C}$ be the set of cubes that intersect $S_{1}$. We partition $\mathcal{C}$ into two sets:
\begin{itemize}
    \item $\mathcal{C}_1 = \{c\in \mathcal{C}: \vol(c \cap S_1) \leq \frac{2}{3} \vol(c)\}$, the subset of cubes in $\mathcal{C}$ where $S_{1}$
takes up at most $\frac{2}{3}$ of the volume
of the cube, and
    \item $\mathcal{C}_2 = \{c\in \mathcal{C}: \vol(c \cap S_1) {>} \frac{2}{3} \vol(c)\}$, the subset of cubes in $\mathcal{C}$ where $S_{1}$
takes up more than $\frac{2}{3}$ of each cube.
\end{itemize}
\begin{figure}
    \centering
    \includegraphics[width = 0.75\textwidth]{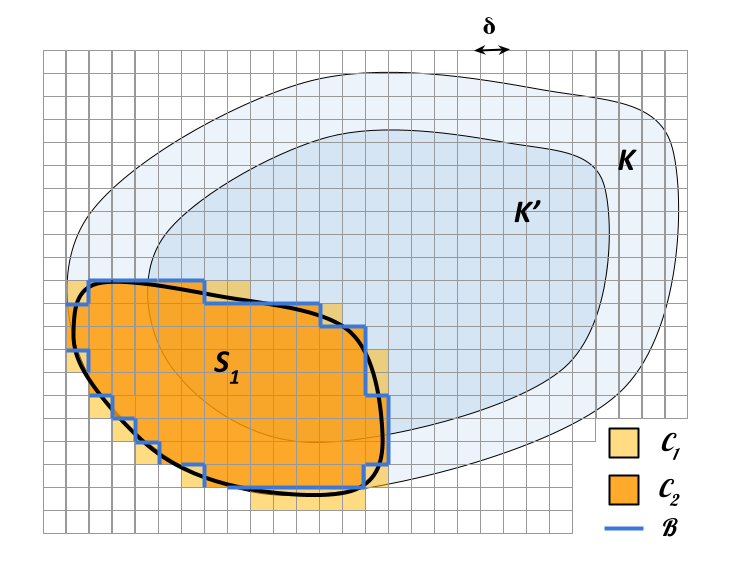}
    \caption{Illustration of the isoperimetry proof. $S_1$ is bounded by the thick black curve. $\mathcal{C}_1$ is shaded in yellow, $\mathcal{C}_2$ is shaded in orange, and $\mathcal{B}$ is colored blue.}
    \label{fig:iso}
\end{figure}

Depending on whether most of the volume $S_1$ is contained in $\mathcal{C}_1$ or $\mathcal{C}_2$, there are two possibilities:

\textbf{Case 1}: If $\vol(\mu({\cal C}_{1})\cap S_{1})\geq \vol(S_{1})/2$,
i.e., at least  half of $\vol(S_{1})$ resides
in cubes in ${\cal C}_{1}$, then we apply Lemma \ref{lem:Cube-isoperimetry} to every cube in $\mathcal{C}_1$ individually to bound $\vol(S_3)$. However, $K$ might not completely contain the cubes in $\mathcal{C}_1$, so before using the cube isoperimetry, we move to the contracted body $K'$. Let
\begin{equation*}
    \mathcal{C}_{1}'=\{c\in{\cal C}_{1}:c\cap K'\neq\emptyset\},
\end{equation*}
i.e., $\mathcal{C}_{1}'$ is the subset of cubes in $\mathcal{C}_1$ that intersect $K'$.
By our choice of $\alpha$, $\mu({\cal C}_{1}')\subseteq K$. Using inequality~\eqref{eq:4} on $\mu({\cal C}_{1}) \cap S_1$ gives
\begin{equation}
   \vol(\mu({\cal C}_{1}')\cap S_{1})\ge\vol(\mu({\cal C}_{1})\cap S_{1}\cap K')\geq\vol(\mu({\cal C}_{1})\cap S_{1})-\frac{\varepsilon}{2} \cdot \vol(K). \label{eq:5}
\end{equation}

Now consider a cube $c$ in ${\cal C}_{1}'$. Using Lemma \ref{lem:cube-} on $c$, 
\begin{align}
   \vol(S_3 \cap c)\geq  \psi_c \cdot \frac{\vol(S_1 \cap c)}{4}. \label{eq:6}
\end{align}
Summing~\eqref{eq:6} over every cube in ${\cal C}_{1}'$, we get 
\begin{align*}
\vol(S_{3}) & \geq\frac{\psi_c}{4}\cdot\sum_{c\in{\cal C}_{1}'}\vol(c\cap S_{1})=\frac{\psi_c}{4}\cdot\vol(\mu({\cal C}_{1}')\cap S_{1})\\
 & \geq\frac{\psi_c}{4}\cdot \left(\vol(\mu({\cal C}_{1})\cap S_{1})-\frac{\varepsilon}{2}\cdot \vol(K)\right) \tag*{(by inequality~\eqref{eq:5})}\\
 & \geq\frac{\psi_c}{4}\cdot \left(\frac{1}{2}\cdot\vol(S_{1})-\frac{\varepsilon}{2}\cdot\vol(K)\right) \tag*{(since $\vol(\mu({\cal C}_{1})\cap S_{1})\geq \vol(S_{1})/2$)}\\
  &=\frac{\psi_c}{8}\cdot \left(\vol(S_{1})-\varepsilon\cdot \vol(K)\right).
\end{align*}

\textbf{Case 2}: If $\vol(\mu({\cal C}_{1})\cap S_{1}) < \frac{1}{2}\cdot \vol(S_{1})$, then $\vol(\mu({\cal C}_{2})\cap S_{1})\geq \frac{1}{2}\cdot\vol(S_{1})$.
Let  $\mathcal{B}$ be the set of facets of the grid-cubes that intersect $\partial_K(\mu({\cal C}_{2})$.
Since ${\cal C}_2$ is a set of axis-aligned $n$-dimensional cubes,  $\mathcal{B}$ is a set of axis-aligned $(n-1)$-dimensional cubes. Consider a facet $f$ in  $\mathcal{B}$ with normal axis $e_{i}$, and let $f_1$ be the grid-cube that is adjacent to $f$ and not in ${\cal C}_{2}$, and $f_2$ be the grid-cube that is adjacent to $f$ an in ${\cal C}_{2}$.

Let $\pi(f)$ denote the projection of $S_1 \cap f_2$ on $f$, i.e.,
\begin{equation*}
    \pi(f) = \{x \in f: \exists y \in S_1 \cap f_2 \;\text{ s.t. }\; x_j = y_j, \; \forall j \in [n]\backslash \{i\}\}.
\end{equation*}
and let $\mathcal{E}(f)$ denote the extension of this projection along the $e_i$-axis in $f_1$, i.e., 
\begin{equation*}
    \mathcal{E}(f) = \{x \in f_1: \exists y \in \pi(f) \;\text{ s.t. }\; x_j = y_j, \; \forall j \in [n]\backslash \{i\}\}.
\end{equation*}
Then $S_2 \cap \mathcal{E}(f) = \emptyset$ as every point in $\mathcal{E}(f)$ is reachable from a point in $S_{1}$ along $e_{i}$ and therefore cannot be in
$S_{2}$. Because the grid size is $\delta$, $\vol_{n-1}(f) \cdot \delta = \vol(f_1) = \vol(f_2)$. 
\begin{figure}[h]
    \centering
    \includegraphics[width = 0.5\textwidth]{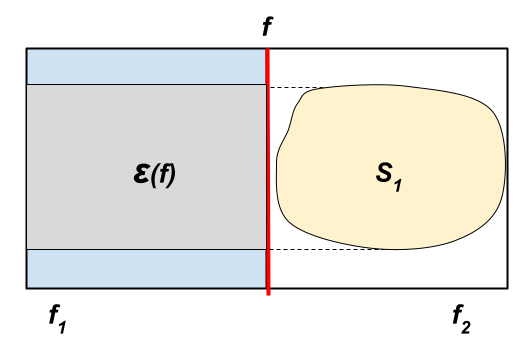}
    \caption{Figure illustrating $\mathcal{E}(f)$ for a facet $f \in \mathcal{B}$. $f$ is represented by the red line. $f_1$ is shaded in blue, $f_2$ is shaded in white, $S_1 \cap f_2$ is shaded in yellow, and $\mathcal{E}(f)$ is shaded in gray.}
    \label{fig:extension}
\end{figure}
\begin{align*}
    \vol_{n-1}(\pi(f)) \cdot \delta &\geq \vol(S_1 \cap f_2) \geq \frac{2}{3}\cdot \vol(f_2) \tag*{(since $f_2 \in \mathcal{C}_2$)}\\
    \Rightarrow \vol_{n-1}(\pi(f)) &\geq  \frac{2}{3}\cdot \frac{\vol(f_2)}{\delta} = \frac{2}{3}\cdot \vol_{n-1}(f).
\end{align*}
This gives 
\begin{equation}
  \vol(\mathcal{E}(f)) \geq \delta\cdot\vol_{n-1}(\pi(f)) \geq  \frac{2}{3}\cdot \delta \cdot \vol_{n-1}(f). \label{eq:7}  
\end{equation}

Since $\mathcal{E}(f)$ does not contain $S_2$, it can contain $S_1$ and $S_3$. If $\mathcal{E}(f)$ does not contain $S_1$, then $\mathcal{E}(f)\subseteq S_3$ and we can add the mass from~\eqref{eq:7} to $S_3$. 
Otherwise $f_1 \in \mathcal{C}_1$, and there are 2 possibilities:
\begin{itemize}
    \item $\vol(S_1  \cap f_1 )\leq \vol(f_1)/3$: We can simply subtract this volume from $\vol(\mathcal{E}(f))$ to get 
    \begin{align*}
        \vol(S_3  \cap f_1 ) &\geq \vol(\mathcal{E}(f)) - \vol(S_1  \cap  f_1) \\
        &\geq
        \frac{2}{3}\cdot \delta\cdot\vol_{n-1}(f) - \frac{1}{3} \cdot \vol(f_1) =\frac{\delta}{3} \cdot \vol_{n-1}(f).
    \end{align*}
    \item $\frac{1}{3}\cdot \vol(f_1) < \vol(S_1  \cap f_1 ) \leq \frac{2}{3}\cdot \vol(f_1)$:  Since $f_1 \in \mathcal{C}_1'$, using Lemma \ref{lem:cube-} on $f_1$, we get
\begin{align*}
   \vol(S_3 \cap f_1) \geq \psi_c \cdot \frac{\vol(S_1\cap f_1)}{4} \geq \frac{\psi_c}{12} \cdot \vol(f_1) = \frac{\psi_c }{12}\cdot \delta\cdot \vol_{n-1}(f).
\end{align*}
\end{itemize}
So, for every facet $f \in \mathcal{B}$, 
\begin{equation}
    \vol(S_3 \cap \mathcal{E}(f)) \geq \min\left\{\frac{\psi_c}{12}, \frac{1}{3} \right\} \cdot \delta \cdot \vol_{n-1}(f)  = \frac{\psi_c}{12}\cdot \delta \cdot \vol_{n-1}(f).
    \label{eq:8}
\end{equation}
Since an $n$-dimensional cube has $2n$ facets, it can contribute to the extension of a facet at most $2n$ times.  Therefore, every facet on the boundary $\mathcal{B}$ contributes at least $\frac{1}{2n} \cdot \frac{\psi_c }{12}\cdot \delta\cdot \vol_{n-1}(f)$ to $\vol(S_3)$. However, $f_{1}$ (and $\mathcal{E}(f)$) might not be fully contained in $K$.
So we again move to the contraction $K'$. Let 
\begin{equation*}
\mathcal{C}_{2}'=\{c\in{\cal C}_{2}:c\cap K'\neq\emptyset\}.    
\end{equation*}

Our choice of $\alpha$ ensures that every cube in $\mathcal{C}_{2}'$ and all the cubes that are adjacent to a cube in $\mathcal{C}_{2}'$ are fully contained
in $K$. Let $I = K' \cap \mu({\cal C}_2)$. Let $\mathcal{B}'$ be the set of facets of grid-cubes that intersect $\partial_{K'}(I)$, the internal boundary of $I$ relative to $K'$. 

\begin{figure}[h]
    \centering
    \includegraphics[width = 0.75\textwidth]{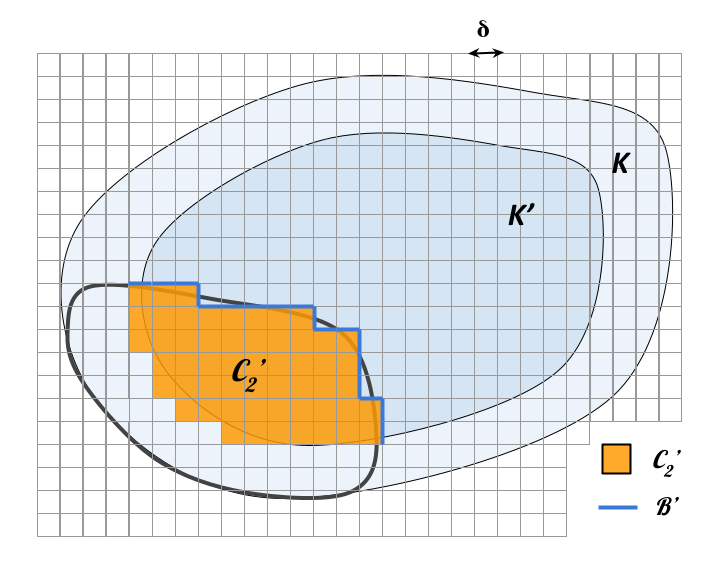}
    \caption{Illustration of $\mathcal{B}'$ and $\mathcal{C}_2'$ for the same convex body $K$ and subset $S_1$ as figure \ref{fig:iso}. Note that $\mathcal{B}' \subseteq \mathcal{B}$ and all the cubes adjacent to any facet in $\mathcal{B}'$ are completely contained in $K$.}
    \label{fig:contraction}
\end{figure}
Then $\mathcal{B}'\subseteq \mathcal{B}$.
To see this, consider a facet $f\in \mathcal{B}'$. Let $f_1, f_2$ be the grid-cubes adjacent to $f$ such that $f_2$ is fully contained in $I$. Because $f$ intersects the internal boundary of $I$ with respect to $K'$, $f_1$ cannot be fully contained in $I$. Since $f_1$ has a neighbor (namely $f_2$) which intersects $K'$, $f_1 \subset K'$. If $f_1 \in \mathcal{C}_2$, then by definition $f_1 \subset I$, which contradicts the fact that $f$ lies on the boundary of $I$. Therefore, $f_1 \notin \mathcal{C}_2$, and as a result $f$ lies on the boundary of $\mathcal{C}_2$, i.e.,  $f \in \mathcal{B}$.

Therefore using~\eqref{eq:8}, every facet $f\in \mathcal{B}'$ contributes at least
\begin{equation*}
    \frac{1}{2n} \cdot \frac{\psi_c}{12} \cdot \delta \cdot \vol_{n-1}(f)
\end{equation*} to $\vol(S_3)$. Summing this up over $\mathcal{B}'$,
\begin{align}
    \vol(S_3) \geq  \frac{1}{2n} \cdot \frac{\psi_c}{12} \cdot \delta \cdot \sum_{f \in \mathcal{B}'}\vol_{n-1}(f) \geq \frac{\delta \cdot \psi_c}{24 n} \cdot \vol_{n-1}\left(\partial_{K'}(I)\right). \label{eq:9}
\end{align}
Using Lemma~\ref{Euclidean-isoperimetry} on $I$ in $K'$,
\begin{equation}
    \vol_{n-1}\left(\partial_{K'}(I)\right)  \geq\frac{\ln{2}}{R}\cdot\min\left\{ \vol\left(I\right),\vol\left(K'\backslash I\right)\right\}. \label{eq:10}
\end{equation}
Using inequality~\eqref{eq:4} on $I$,
\begin{align}
    \vol(I) &= \vol(K' \cap \mu(\mathcal{C}_2)) \geq \vol(\mu(\mathcal{C}_2)) - \frac{\varepsilon}{2}\cdot \vol(K)\notag\\
    &\geq \frac{1}{2} \cdot (\vol(S_1) - \varepsilon \cdot \vol(K)), \label{eq:11}
\end{align}
where the last inequality follows from $\vol(\mu({\cal C}_{2})\cap S_{1})\geq \frac{1}{2}\cdot\vol(S_{1})$.
On the other hand, since $I \subseteq \mathcal{C}_2'$,
\begin{equation}
    \vol(I) \leq \vol\left(\mu({\cal C}_{2}')\right)\leq\frac{3}{2} \cdot \vol\left(S_{1}\cap \mu({\cal C}_{2}')\right)\leq\frac{3}{4}\cdot \vol\left(K\right), 
    \label{eq:12}
\end{equation}
where the second inequality is because $S_1$ occupies at least $\frac{2}{3}$ of every cube in $\mathcal{C}_2'$.
Therefore,
\begin{align}
     \vol(K' \backslash I ) &\geq \vol(K') - \frac{3}{4}\cdot \vol(K)\geq \left(1-\frac{\varepsilon}{2} -\frac{3}{4}\right)\cdot \vol(K) \notag \\
     &\geq \frac{\vol(K)}{4}-\frac{\varepsilon}{2} \cdot \vol(K)\notag\\
    &\geq \frac{1}{2}\cdot\left( \vol(S_1) - \varepsilon \cdot \vol(K) \right), \label{eq:13}
\end{align}
where the last inequality uses the fact that $\vol(S_1)\leq \frac{1}{2}\cdot \vol(K)$.
Combining inequalities~\eqref{eq:10},~\eqref{eq:11}, and~\eqref{eq:13}, we get
\begin{align}
\vol_{n-1}\left(\partial_{K'}(I)\right) 
 & \geq\frac{\ln{2}}{R}\cdot\frac{1}{2}\cdot\left( \vol(S_1) - \varepsilon\cdot \vol(K) \right) \label{eq:14}.
\end{align}
Plugging the bound from inequality~\eqref{eq:14} into inequality~\eqref{eq:9} gives
\begin{align*}
\vol(S_{3}) & \geq \frac{ \delta \cdot \psi_c }{24 n } \cdot \vol_{n-1}(\partial_{K'}(I)) \\
&\geq\frac{\delta\cdot \psi_c \cdot \ln{2}}{48 R n} \cdot\left( \vol(S_1) - \varepsilon \cdot\vol(K) \right)
\end{align*}
Using $\delta = \frac{\alpha}{4\sqrt{n}} = \frac{\varepsilon}{8 n\sqrt{n}}$ and $\psi_c = \frac{\ln{2}}{n}$, we have 
\begin{equation*}
    \vol(S_{3})\ge \frac{\varepsilon}{800 \cdot R n^{3.5}} \cdot\left( \vol(S_1) - \varepsilon \cdot \vol(K) \right). \qedhere
\end{equation*} 
\end{proof}

\section{Conductance}
\label{sec:conductance}
In this section, we bound the $s$-conductance of CHAR. The following simple
lemma lets us reduce the $s$-conductance of $K$ to the isoperimetry of axis-disjoint subsets of $K$. 
\begin{lemma}
\label{lem:disjoint}Let $S_1\subseteq K$ be a measurable subset of $K$ and $S_2 = K \backslash S_1$. Let $S_{1}'=\{x\in S_{1}:P_{x}(S_{2})<\frac{1}{2n}\}$ and $S_{2}'=\{x\in S_{2}:P_{x}(S_{1})<\frac{1}{2n}\}$.
Then $S_{1}'$ and $S_{2}'$ are axis disjoint.
\end{lemma}

\begin{proof}
For the sake of contradiction, assume that $S_1'$ and $S_2'$ are not axis-disjoint. Then there exists an axis-parallel line, $\ell$, passing through
both $S_{1}'$ and $S_{2}'$.

Let $x\in S_{1}'\cap \ell$ and $y\in S_{2}'\cap \ell$. From the definition of $S_1'$ and $S_2'$,
\begin{align}
     \frac{1}{n}\dfrac{\len(\ell\cap S_{2})}{\len(\ell \cap K)} &\leq P_{x}(S_{2})\leq \frac{1}{2n}, \;\text{and}  \label{eq:15} \\
      \frac{1}{n}\dfrac{\len(\ell\cap S_{1})}{\len(\ell \cap K)} &\leq P_{y}(S_{1})\leq \frac{1}{2n} \;.\label{eq:16}
\end{align}
Adding~\eqref{eq:15} and~\eqref{eq:16} implies that 
\begin{equation*}
    \len(\ell\cap S_{2}) + \len(\ell\cap S_{1})<\len(\ell\cap K),
\end{equation*}
which is a contradiction. Therefore, $S_1'$ and $S_2'$ are axis-disjoint.
\end{proof}
\begin{theorem} \label{thm:conductance}
Let $K$ be a convex body in $\mathbb{\mathbb{R}}^{n}$ containing
a unit ball with $R^{2}=\E_{K}(\norm{x-z_K}^{2})$ where $z_K$ is the centroid of $K$. Then the $s$-conductance
of Coordinate Hit-and-Run in $K$ is at least $\dfrac{s}{128\cdot 10^2 \cdot R n^{4.5}}$.
\end{theorem}

\begin{proof}
Let $S_1 \subseteq K$ be a measurable subset of $K$ with $s < \pi_K(S_1)\leq 1/2$ and let $S_2 = K\backslash S_1$.
Let 
\begin{equation*}
 S_{1}'=\{x\in S_{1}:P_{x}(S_{2})<\frac{1}{2n}\}
\quad
\text{ and }
\quad
S_{2}'=\{x\in S_{2}:P_{x}(S_{1})<\frac{1}{2n}\}.
\end{equation*}
The ergodic flow of $S_1$ is given by
\begin{align*}
p(S_1) &= \int_{x\in S_{1}}P_{x}(S_{2})\; {d\pi_K(x)} \\
& =\int_{x\in S_{1}'}P_{x}(S_{2})\;d\pi_K(x)+\int_{x\in S_{1}\backslash S_{1}'}P_{x}(S_{2})\;d\pi_K(x).
\end{align*} 
Since $\int_{x\in S_{1}\backslash S_{1}'}P_{x}(S_{2})\;d\pi_K(x) \geq \vol(S_{1}\backslash S_{1}')/2n$, we get
\begin{equation}
    p(S_1) \geq \int_{x\in S_{1}\backslash S_{1}'}P_{x}(S_{2})\;d\pi_K(x) \geq \frac{\vol(S_{1}\backslash S_{1}')}{2n \cdot \vol(K)}\;. \label{eq:17}
\end{equation}
We can also expand the ergodic flow of $S_1$ as
\begin{align}
p(S_1) = \int_{x\in S_{1}}P_{x}(S_{2}) &\;d\pi_K(x) \geq\int_{x\in S_{1}}P_{x}(S_{2}\backslash S_{2}') \;d\pi_K(x) \notag \\
&\stackrel{(1)}{=}\int_{y\in S_{2}\backslash S_{2}'}P_{y}(S_{1}) \;d\pi_K(y)\geq\frac{\vol(S_{2}\backslash S_{2}')}{2n \cdot \vol(K)}\;,\label{eq:18}
\end{align} 
where (1) follows from time-reversibility of the Markov chain.

If $\vol(S_{1}')<\vol(S_{1})/2$, then by~\eqref{eq:17}, \begin{equation*}
    p(S_1)  \geq \frac{\vol(S_{1})}{4n \cdot \vol(K)} = \frac{\pi_K(S)}{4n}\;.
\end{equation*}
This implies $\phi_s(S_1) \geq \frac{1}{4n}$. Similarly, if $\vol(S_{2}')<\vol(S_{2})/2$, then by \eqref{eq:18},
\begin{equation*}
     p(S_1)  \geq \frac{\vol(S_{2})}{4n \cdot \vol(K)} \geq \frac{\vol(S_{1})}{4n \cdot \vol(K)} = \frac{\pi_K(S)}{4n}\;.
\end{equation*}
This also implies $\phi_s(S_1) \geq \frac{1}{4n}$. 

So, assume that $\vol(S_{1}')\geq \vol(S_{1})/2$ and $\vol(S_{2}')\geq\vol(S_{2})/2$. Let $S_{3}'=K\backslash \{ S_{1}'\cup S_{2}'\}$. Lemma~\ref{lem:disjoint} implies that $S_{1}'$ and $S_{2}'$ are axis-disjoint. Thus, using
Theorem~\ref{thm:iso} with $\varepsilon = s/2$, we get
\begin{equation}
    \vol(S_{3}')\geq s\cdot \psi \cdot \left(\min\{\vol(S_{1}'),\vol(S_{2}')\}-\frac{s}{2}\cdot\vol(K)\right),  \label{eq:19}
\end{equation}
where $\psi = 1/(1600 \cdot n^{3.5}R)$.

Adding~\eqref{eq:17} and~\eqref{eq:18},
 \begin{align*}
p(S_1)&\geq\frac{1}{2}\cdot \left(\frac{\vol(S_{1}\backslash S_{1}')}{2n \cdot \vol(K)} + \frac{\vol(S_{2}\backslash S_{2}')}{2n \cdot \vol(K)}\right)=\frac{\vol(S_{3}')}{4n \cdot \vol(K)} \\
 &\geq \frac{s\cdot \psi}{4n \cdot \vol(K)}\cdot\left(\min\{\vol(S_{1}'),\vol(S_{2}')\}-\frac{s}{2}\cdot\vol(K)\right) \tag*{(from~\eqref{eq:19})}\\
 &\geq\frac{s\cdot \psi}{8n \cdot \vol(K)} \cdot\left(\min\{\vol(S_{1}),\vol(S_{2})\}- s\cdot\vol(K)\right)\\
 & \geq\frac{s \cdot \psi }{8n \cdot \vol(K)}\cdot \left(\vol(S_{1})- s \cdot\vol(K)\right) = \frac{s \cdot \psi}{8n}\cdot \left(\pi_K(S_{1})- s\right).
\end{align*}
So, for any $S_1 \subseteq K$ with $s < \pi_K(S_1) \leq 1/2$, we get
\begin{equation*}
    \frac{p(S_1)}{\pi_K(S_1)-s}\geq \frac{s \cdot \psi}{8n}.
\end{equation*}
Thus, the $s$-conductance of the CHAR, $\phi_s$, is at least
\begin{equation*}
    \phi_s \geq \frac{s \cdot \psi}{8n} = \frac{s}{128\cdot 10^2 \cdot R n^{4.5}}. \qedhere
\end{equation*}
\end{proof}

\begin{proof}[Proof of Theorem \ref{thm:mixing}]
For a convex body $K$, let $\pi_0$ be the starting distribution of CHAR such that $\pi_0$ is $M$-warm with respect to $\pi_K$. Let $\pi_t$ be the distribution after $t$ steps of CHAR. Using Theorem \ref{thm:sconductance} with $s = \frac{\varepsilon}{2M}$, we get
\begin{align*}
    H_s \leq \frac{\varepsilon}{2}  \quad \text{and} \quad d_{TV}(\pi_{t},\pi_K)\leq \frac{\varepsilon}{2} + M\left(1-\frac{\phi_{s}^{2}}{2}\right)^{t}.
\end{align*}
The above inequality, along with Theorem \ref{thm:conductance}, implies that
\begin{align*}
t = \frac{4}{\phi_s^2} \log \frac{2M}{\varepsilon} <   7\cdot 10^4 \cdot \frac{ M^2 R^2 n^9}{\varepsilon^2} \log \frac{2M}{\varepsilon}
\end{align*}
steps of the CHAR suffice to ensure $ d_{TV}(\pi_t,\pi_K)\leq \varepsilon$.
\end{proof}

\section{\label{sec:Lower-bound}Lower bound}
\begin{theorem}
There exists a convex body $K \in \mathbb{R}^n$ containing a unit ball with diameter $D$, such that the conductance of CHAR on $K$ is $O(1/n^{1.5}D)$.
\end{theorem}

We construct a convex body $K \in \mathbb{R}^n$ containing a unit ball such that $K$ is a skew-cylinder along the $e_1$ axis whose cross-sections are $(n-1)$-dimensional unit balls. The skewness of the body forces CHAR to take small steps in the $e_1$ direction.

Formally, let $B(x)$ be an $(n-1)$-dimensional unit ball centered at $(x,0,\ldots, 0)$. The convex body $K$ is defined as
\begin{equation*}
    K = \{x \in \mathbb{R}^n: 0 \leq x_1 \leq D, \; (x_2, \ldots, x_n) \in B(x_1) \}.
\end{equation*}

\begin{figure}
    \centering
    \includegraphics[width = 0.55\textwidth]{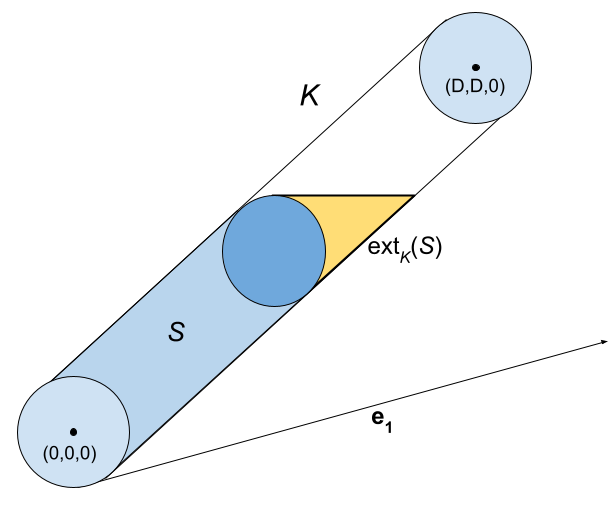}
    \caption{The lower bound construction in 3 dimensions. $S$ is shaded in blue. The extension of $S$ along $e_1$ is shaded in yellow. To escape $S$, Coordinate Hit-and-Run needs to move along the $e_1$ axis. However, the skewness of $K$ forces the step size along $e_1$ to be small.}
    \label{fig:skew}
\end{figure}

Let $d = D/2$ and let $S = \{x \in K: x_1 \leq d\}$. Then, we claim that $\phi(S)$ is $O\left(1/n^{1.5}D\right)$. Before proving this claim, we define axis-parallel extensions of subsets of convex bodies.

\begin{definition}[Axis-parallel Extension]
For a convex body $K$ in $\mathbb{R}^n$ and a measurable subset $S \subseteq K$, the axis-parallel extension of $S$ in $K$, denoted by $\text{ext}_K(S)$ is defined as 
\begin{equation*}
    \mathrm{ext}_K(S) = \{x \in K\backslash S: \exists y \in S \text{ such that } \vert \{ i \in [n]: x_i = y_i\}\vert = n-1\}.
\end{equation*}
In other words, $\text{ext}_K(S)$ is the set of points in $K\backslash S$ obtained by changing exactly $1$ coordinate from a point in $S$. 
\end{definition}

We will first bound the volume of $\mathrm{ext}_K(S)$ and then use it to bound the conductance of CHAR on $K$.
\begin{lemma}
The volume of the extension of $S$ in $K$ is at most
\begin{equation*}
    \vol(\mathrm{ext}_K(S)) \leq \frac{40}{\sqrt{n}D} \cdot \vol(S).
\end{equation*}
\end{lemma}
\begin{proof}
The extension of $S$ goes beyond $S$ only along the $e_1$ axis, and 
\begin{equation*}
    \mathrm{ext}_K(S) = \{x \in K: x_1 \in (d, d+1], \; (x_2, \ldots, x_n) \in B(d) \cap B(x_1)\}.
\end{equation*}
Each $(n-1)$-dimensional slice of $ \mathrm{ext}_K(S)$ is the intersection of two $(n-1)$-dimensional unit balls. This gives 
\begin{align*}
    \vol( \mathrm{ext}_K(S)) &= \int_{t = 0}^1 \vol(B(d) \cap B(d+t)) dt. 
\end{align*}

\begin{figure}[h]
  \begin{center}
    \includegraphics[scale=0.4]{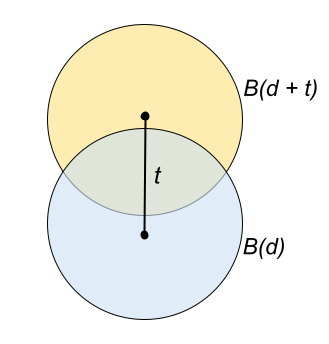}
  \end{center}
  \caption{$B(d+t) \cap B(t)$}
  \label{fig:intersection}
\end{figure}
Note that $\vol(B(d) \cap B(d+t))$ is equal to twice the volume of spherical cap at distance $t/2$ from the center of an $(n-1)$ unit ball. 
By \cite{tkocz2012upper}, the volume of a spherical cap at distance $t$ from the center of an $(n-1)$-dimensional unit ball, $C_{n-1}$, is at most $e^{-\frac{nt^2}{2}} \cdot \vol(C_{n-1})$.

Thus $\vol(B(d) \cap B(d+t)) \leq e^{-\frac{nt^2}{8}} \cdot \vol_{n-1}(B(d))$, and
\begin{align*}
    \vol( \mathrm{ext}_K(S)) &= \int_{t = 0}^1 \vol(B(d) \cap B(d+t)) dt \\
    &\leq 2\cdot \int_{t = 0}^\infty e^{-\frac{nt^2}{8}} \cdot \vol_{n-1}(B(d))dt \leq \frac{20}{\sqrt{n}} \cdot \vol_{n-1}(B(d)).
\end{align*}
Since $\vol(S) = \frac{D}{2} \cdot \vol_{n-1}(B(d))$, we have $\vol( \mathrm{ext}_K(S)) \leq \frac{40}{\sqrt{n}D}\cdot \vol(S)$.
\end{proof}

\begin{lemma}
The conductance of Coordinate Hit-and-Run on $K$ is $O(1/n^{1.5}D)$.
\end{lemma}
\begin{proof}
The only points in $S$ with non-zero probability of escaping in one step are
\begin{equation*}
    S_1 = \left\{x: d-1\leq x_1 \leq d, \; (x_2, \ldots, x_n) \in B(d) \cap B(x_1)\right\}.
\end{equation*} 
Any point in $S_1$ can move out of $S$ in one step of CHAR if only if the first coordinate is selected for re-sampling. Therefore, for any $x \in S_1$,
\begin{equation*}
    P_x(K\backslash S) \leq 1/n.
\end{equation*}
By symmetry, $\vol(S_1) = \vol(\mathrm{ext}_K(S)) \leq \frac{40}{\sqrt{n} D} \cdot \vol(S)$.
Therefore,
\begin{align*}
    p(S) = \int_S P_x(K\backslash S) &\;d\pi_K(x) =  \int_{S_1} P_x(K\backslash S) \;d\pi_K(x) \\
    &\leq  \frac{1}{n} \cdot \frac{\vol(S_1)}{\vol(K)} \leq \frac{1}{n} \cdot \frac{40}{\sqrt{n}D} \cdot \pi_K(S).
\end{align*}
Therefore, $\phi(S) \leq 40/(n^{1.5}D)$, and the conductance of CHAR in $K$, $\phi$, is at most $O(1/n^{1.5}D)$. 
\end{proof}

We expect that this translates to a lower bound of $\Omega^*(n^{3}D^{2})$
on the mixing rate even from a warm start. Even though
this is worse than the $O^*(n^{2}D^{2})$ mixing rate of
Hit-and-Run, it is an interesting open problem to determine the precise
mixing rate of CHAR.

\vspace{1em}
\noindent\textbf{Acknoledgements}:
This work was supported in part by NSF awards DMS-1839323, CCF-1909756 and CCF-2007443.
The authors thank Ben Cousins for helpful discussions.
\bibliographystyle{spmpsci}
\bibliography{char}

\begin{thebibliography}{10}
\providecommand{\url}[1]{{#1}}
\providecommand{\urlprefix}{URL }
\expandafter\ifx\csname urlstyle\endcsname\relax
  \providecommand{\doi}[1]{DOI~\discretionary{}{}{}#1}\else
  \providecommand{\doi}{DOI~\discretionary{}{}{}\begingroup
  \urlstyle{rm}\Url}\fi

\bibitem{Anderson2007}
Andersen, H.C., Diaconis, P.: Hit and run as a unifying device.
\newblock Journal de la soci\'et\'e fran\c caise de statistique
  \textbf{148}(4), 5--28 (2007)

\bibitem{Boneh83}
Boneh, A.: Preduce --- a probabilistic algorithm identifying redundancy by a
  random feasible point generator (rfpg).
\newblock In: Redundancy in Mathematical Programming, pp. 108--134. Springer
  Berlin Heidelberg, Berlin, Heidelberg (1983)

\bibitem{CV13}
Cousins, B., Vempala, S.: Volume computation of convex bodies.
\newblock MATLAB File Exchange  (2013).
\newblock
  Http://www.mathworks.com/matlabcentral/fileexchange/\\43596-volume-computation-of-convex-bodies

\bibitem{CV2015}
Cousins, B., Vempala, S.: Bypassing {KLS}: {Gauss}ian cooling and an
  {$O^*(n^3)$} volume algorithm.
\newblock In: STOC, pp. 539--548 (2015)

\bibitem{CV15b}
Cousins, B., Vempala, S.: A practical volume algorithm.
\newblock Mathematical Programming Computation \textbf{8}(2), 133--160 (2016)

\bibitem{diaconis2010gibbs}
Diaconis, P., Khare, K., Saloff-Coste, L.: Gibbs sampling, conjugate priors and
  coupling.
\newblock Sankhya A \textbf{72}(1), 136--169 (2010)

\bibitem{diaconis2012gibbs}
Diaconis, P., Lebeau, G., Michel, L.: Gibbs/metropolis algorithms on a convex
  polytope.
\newblock Mathematische Zeitschrift \textbf{272}(1-2), 109--129 (2012)

\bibitem{emiris2014efficient}
Emiris, I.Z., Fisikopoulos, V.: Efficient random-walk methods for approximating
  polytope volume.
\newblock In: Proceedings of the thirtieth annual symposium on Computational
  geometry, pp. 318--327 (2014)

\bibitem{finkel2005incorporating}
Finkel, J.R., Grenager, T., Manning, C.D.: Incorporating non-local information
  into information extraction systems by gibbs sampling.
\newblock In: Proceedings of the 43rd Annual Meeting of the Association for
  Computational Linguistics (ACL'05), pp. 363--370 (2005)

\bibitem{DBLP:journals/pami/GemanG84}
Geman, S., Geman, D.: Stochastic relaxation, gibbs distributions, and the
  bayesian restoration of images.
\newblock {IEEE} Trans. Pattern Anal. Mach. Intell. \textbf{6}(6), 721--741
  (1984).
\newblock \doi{10.1109/TPAMI.1984.4767596}.
\newblock \urlprefix\url{https://doi.org/10.1109/TPAMI.1984.4767596}

\bibitem{george1993variable}
George, E.I., McCulloch, R.E.: Variable selection via gibbs sampling.
\newblock Journal of the American Statistical Association \textbf{88}(423),
  881--889 (1993)

\bibitem{kannan2003rapid}
Kannan, R.: Rapid mixing in markov chains.
\newblock arXiv preprint math/0304470  (2003)

\bibitem{KLS95}
Kannan, R., Lov{\'a}sz, L., Simonovits, M.: Isoperimetric problems for convex
  bodies and a localization lemma.
\newblock Discrete \& Computational Geometry \textbf{13}, 541--559 (1995)

\bibitem{KLS97}
Kannan, R., Lov\'{a}sz, L., Simonovits, M.: Random walks and an {$O^*(n^5)$}
  volume algorithm for convex bodies.
\newblock Random Structures and Algorithms \textbf{11}, 1--50 (1997)

\bibitem{LeeV17KLS}
Lee, Y.T., Vempala, S.S.: Eldan's stochastic localization and the {KLS}
  hyperplane conjecture: An improved lower bound for expansion.
\newblock In: Proc. of IEEE FOCS (2017)

\bibitem{loomis1949inequality}
Loomis, L.H., Whitney, H., et~al.: An inequality related to the isoperimetric
  inequality.
\newblock Bulletin of the American Mathematical Society \textbf{55}(10),
  961--962 (1949)

\bibitem{L90}
Lov\'{a}sz, L.: How to compute the volume?
\newblock Jber. d. Dt. Math.-Verein, Jubil\"aumstagung 1990 pp. 138--151 (1990)

\bibitem{Lovasz1998}
Lov\'asz, L.: Hit-and-run mixes fast.
\newblock Math. Prog. \textbf{86}, 443--461 (1998)

\bibitem{LS93}
Lov\'asz, L., Simonovits, M.: Random walks in a convex body and an improved
  volume algorithm.
\newblock In: Random Structures and Alg., vol.~4, pp. 359--412 (1993)

\bibitem{LV3}
Lov\'{a}sz, L., Vempala, S.: Hit-and-run from a corner.
\newblock SIAM J. Computing \textbf{35}, 985--1005 (2006)

\bibitem{LV2}
Lov\'{a}sz, L., Vempala, S.: Simulated annealing in convex bodies and an
  {$O^*(n^4)$} volume algorithm.
\newblock J. Comput. Syst. Sci. \textbf{72}(2), 392--417 (2006)

\bibitem{LV1}
Lov\'{a}sz, L., Vempala, S.: The geometry of logconcave functions and sampling
  algorithms.
\newblock Random Struct. Algorithms \textbf{30}(3), 307--358 (2007).
\newblock \doi{10.1002/rsa.v30:3}

\bibitem{narayanan2020mixing}
Narayanan, H., Srivastava, P.: On the mixing time of coordinate hit-and-run.
\newblock arXiv preprint arXiv:2009.14004  (2020)

\bibitem{SJ89}
Sinclair, A., Jerrum, M.: Approximate counting, uniform generation and rapidly
  mixing {Markov} chains.
\newblock Information and Computation \textbf{82}, 93--133 (1989)

\bibitem{Smith84}
Smith, R.: Efficient {Monte-Carlo} procedures for generating points uniformly
  distributed over bounded regions.
\newblock Operations Res. \textbf{32}, 1296--1308 (1984)

\bibitem{tkocz2012upper}
Tkocz, T.: An upper bound for spherical caps.
\newblock The American Mathematical Monthly \textbf{119}(7), 606--607 (2012)

\bibitem{Turchin1971}
Turchin, V.: On the computation of multidimensional integrals by the
  monte-carlo method.
\newblock Theory of Probability \& Its Applications \textbf{16}(4), 720--724
  (1971).
\newblock \doi{10.1137/1116083}.
\newblock \urlprefix\url{https://doi.org/10.1137/1116083}

\end{thebibliography}
\end{document}